\RequirePackage{amsmath}
\documentclass[a4paper]{llncs}

\usepackage{comment}
\usepackage{amsmath}
\usepackage{amssymb}
\usepackage{graphicx}
\usepackage{cite}
\usepackage[mathlines]{lineno}
\usepackage[linesnumbered,lined,ruled,noend]{algorithm2e}

\let\doendproof\endproof
\renewcommand\endproof{~\hfill$\qed$\doendproof}

\author{%
  Kunihiro Wasa\inst{1} \and Takeaki Uno\inst{1}
}
\institute{%
  National Institute of Informatics, Tokyo, Japan\\
  \texttt{\{wasa, uno\}@nii.ac.jp}
}

\newcommand{\name}[1]{\emph{#1}}

\newcommand{\set}[1]{\left\{#1\right\}}
\newcommand{\inset}[2]{\left\{#1 \;\middle|\; #2\right\}}
\newcommand{\len}[1]{\left|#1\right|}
\newcommand{\size}[1]{\len{#1}}
\newcommand{\tuple}[1]{\left(#1\right)}

\newcommand{\Order}[1]{\mathcal{O}\tuple{#1}}

\makeatletter
\@ifpackageloaded{enumitem}{
    \setlist[enumerate,1]{label=(\Roman*), leftmargin=*}
    \setlist[enumerate,2]{label=(\Roman{enumi}.\alph*), leftmargin=*}
    \setlist[enumerate,3]{label=(\Roman{enumi}.\alph{enumii}.\arabic*), leftmargin=*}
}{

}
\makeatother

\makeatletter

\newcommand{\dist}[1]{dist\tuple{#1}}

\newcommand{\Lg}[1]{L\tuple{#1}}
\newcommand{\Rg}[1]{R\tuple{#1}}

\newcommand{\NE}[1]{N'\tuple{#1}}
\newcommand{\NEc}[1]{N'_c\tuple{#1}}

\newcommand{\sSOL}[1]{\mathcal{B}^{E}\tuple{#1}}
\newcommand{\iSOL}[1]{\mathcal{B}^{V}\tuple{#1}}

\newcommand{\Addv}[3]{\Gamma\tuple{#1, #2, #3}}
\newcommand{\CAND}[2]{\mathit{C}\tuple{#1, #2}}
\newcommand{\CANDR}[2]{\mathit{C}_R\tuple{#1, #2}}
\newcommand{\CANDL}[2]{\mathit{C}_L\tuple{#1, #2}}

\newcommand{\CANDE}[2]{\mathit{C}'\tuple{#1, #2}}

\newcommand{\DELE}[1]{\mathit{Del}'\tuple{#1}}

\newcommand{\LN}[1]{\mathit{LN}\tuple{#1}}
\newcommand{\VSN}[1]{\mathit{VSN}\tuple{#1}}
\newcommand{\VSNc}[1]{\mathit{VSN_C}\tuple{#1}}
\newcommand{\VSNnc}[1]{\mathit{VSN_{\overline{C}}}\tuple{#1}}
\newcommand{\USN}[1]{\mathit{USN}\tuple{#1}}

\newcommand{\FlagInName}{\textsc{FlagIn}}
\newcommand{\FlagIn}[1]{\FlagInName\tuple{#1}}

\newcommand{\InC}{\textsc{InC}}
\newcommand{\NotInC}{\textsc{NotInC}}

\SetKwProg{Procedure}{Procedure}{}{}
\SetKwProg{Subprocedure}{Subprocedure}{}{}
\SetKwComment{tcc}{// }{}
\SetKwFunction{Output}{output}%
\SetKwInOut{AlgInput}{input}
\SetKwInOut{AlgOutput}{output}

\SetKwFunction{Main}{Main}%
\SetKwFunction{Rec}{Rec}%
\SetKwFunction{ComputeChildGenerator}{ComputeChildGen}%

\title{Efficient Enumeration of Bipartite Subgraphs in Graphs}

\begin{document}
\maketitle

\begin{abstract}
    Subgraph enumeration problems ask to output all subgraphs of an input graph that belongs to the specified graph class or satisfy the given constraint.  
    These problems have been widely studied in theoretical computer science. 
    As far, 
    many efficient enumeration algorithms for the fundamental substructures such as spanning trees, cycles, and paths, 
    have been developed. 
    This paper addresses the enumeration problem of bipartite subgraphs.
    Even though bipartite graphs are quite fundamental and have numerous applications in both theory and application, its enumeration algorithms have not been intensively studied, to the best of our knowledge.
    We propose the first non-trivial algorithms for enumerating all bipartite subgraphs in a given graph. 
    As the main results, 
    we develop two efficient algorithms: 
    the one enumerates all bipartite induced subgraphs of a graph with degeneracy $k$ in 
    $\Order{k}$ time per solution. 
    The other enumerates all bipartite subgraphs in $\Order{1}$ time per solution. 
\end{abstract}

\section{Introduction}

A \name{subgraph enumeration problem} is, for given a graph $G$ and a constraint $R$, to output all subgraphs of $G$ that satisfy $R$ once for each and without duplication.
An example is to enumerate all the trees in the given graph, and all the subgraphs whose minimum degree is at least $k$.
The complexity and polynomiality of the subgraph enumeration have been intensively studied in theoretical computer science in the terms of both output size sensitivity and input size sensitivity.
Compared to optimization approach, enumeration has an advantage on exploring and investigating all possibilities and all aspects of the data, thus is widely studied in a practical point of view,  e.g. Bioinformatics~\cite{Ahmed:Neville:Rossi:Duffield:2015}, machine learning~\cite{Ruggieri:2017}, and data mining~\cite{Uno:2004,Zaki:2000}. 
We say that 
an enumeration algorithm is efficient if the algorithm is \name{output sensitive}~\cite{Johnson1988}. 
Especially, we say that $\mathcal{A}$ runs in \name{polynomial amortized time}, if the total running time of an enumeration algorithm $\mathcal{A}$ is $\Order{N\cdot poly(n)}$ time, 
where $N$ is the number of solutions, $n$ is the size of input, and $poly$ is a polynomial function. 
That is, $\mathcal{A}$ enumerates all solutions in $poly(n)$ time per solution. 
Such algorithms have been considered to be efficient, and one of our research goals is to develop efficient enumeration algorithms.
As far, 
there have been studied enumeration algorithms for many fundamental graph structures such as 
spanning trees~\cite{Read:Tarjan:1975,Shioura:Tamura:Uno:1997}, 
$st$-paths~\cite{Read:Tarjan:1975,Ferreira:Grossi:Rizzi:2011}, 
cycles~\cite{Read:Tarjan:1975,Birmele:2012,Ferreira:2014}, 
maximal cliques~\cite{Makino2004,Eppstein:2011,Conte:Grossi:Marino:Versari:2016}, 
minimal dominating sets~\cite{Kante2011}, 
and so on. 
See the comprehensive list in~\cite{Wasa:2016} of this area. 
Recently, Uno~\cite{Uno:2015} developed a technique for a fine-grained analysis of enumeration algorithms. 

Bipartite graph is a well-known fundamental graph structure. 
A bipartite graph is a graph containing no cycle of odd length, that is, 
whose vertex set can be partitioned into two disjoint independent sets. 
Bipartite graphs widely appears in real-world graphs such as 
itemset mining~\cite{Zaki:2000,Uno:2004}, 
chemical information~\cite{Koichi:Arisaka:Koshino:Aoki:Iwata:Uno:Satoh:2014}, 
Bioinformatics~\cite{Zhang2014}, 
and so on. 
Further, enumeration problems for matchings~\cite{Gely:Nourine:Sadi:2009,Uno:1997,Fukuda:Matui:1994} and bicqliue~\cite{Makino2004,Dias:Figueiredo:Szwarcfiter:2005} in bipartite graphs are well studied. 
However, 
to the best of our knowledge, 
there has been proposed no non-trivial enumeration algorithm for bipartite subgraphs. 

In this paper, we propose efficient enumeration algorithms for 
bipartite induced subgraphs and bipartite subgraphs. 
For enumerating both substructures, 
we employ a simple binary partition method, and develop a data structure for efficiently updating the candidates that are called \emph{child generators}. 
Intuitively speaking, 
child generators are vertices or edges such that adding them to a current solution generates another solution. 
For bipartite induced subgraph, 
we look at the \name{degeneracy}~\cite{Lick:White:1970} of a graph. 
The degeneracy of a graph is the upper bound of the minimum degree of any its subgraph, so the graph is sparse when the degeneracy is small.
It is a widely considered as a sparsity measure~\cite{Wasa:Arimura:Uno:2014,Conte:Grossi:Marino:Versari:2016,Eppstein2010,Xu2015}.   
There are several graph classes have constant degeneracies, e.g.,  
forests, grid graphs, planar graphs, bounded treewidth graphs, $H$-minor free graphs with some fixed $H$, and so on~\cite{Lick:White:1970}. 
In addition, 
Real-world graphs such as road networks, social networks, and internet networks are said to often have small degeneracies, or do so after removing a small part of vertices.
Our algorithm utilizes a \emph{good} ordering on the vertices called a \name{degeneracy ordering}~\cite{Matula:Beck:1983}, 
that achieves $\Order{k}$ amortized time per solution, where $k$ is the degeneracy of an input graph. 
This implies that when we restrict the class of input graphs, such as planar graphs, 
the algorithm runs in constant time per solution and is optimal in the sense of time complexity. 
Next, 
for developing an algorithm for bipartite induced subgraph, 
we show that we can avoid redundant edge additions and removal to obtain a solution from another solution. 
As a main result, 
we give an optimal enumeration algorithm, 
that is, the algorithm runs in constant time per solution. 
These algorithms are quite simple, but by giving non trivial analysis, 
we show the algorithms are efficient. 
These are the first non-trivial efficient enumeration algorithms for bipartite subgraphs.



\section{Preliminaries}
Let $G = (V, E)$ be an \name{undirected graph} with vertex set $V=\{1,\ldots,n\}$ and edge set $E=\{(e_1,\ldots,e_m\} \subseteq V\times V$. 
An edge is denoted by $e = \tuple{u, v}$.
We say that $u$ and $v$ are \name{endpoints} of $e = \tuple{u, v}$, and 
$u$ is \name{adjacent} to $v$ if $(u, v) \in E$. 
When the graph is undirected, $\tuple{u, v} = \tuple{v, u}$. 
Two edges are said to be adjacent to each other if a vertex is an end point of both edges. 
The set of \name{neighbors} of $v$ is the set of vertices that are adjacent to $v$ and is denoted by $N(v)$. 
For any vertex subset $S$ of $V$, 
$E[S] = E \cap (S \times S)$, 
that is, $E[S]$ is the set of edges whose both endpoints are in $S$.  
For any edge subset $F$ of $E$, 
$V[F] = \inset{v \in V}{\exists e \in F (v \in e)}$, 
that is, $V[F]$ is the collection of endpoints of edges in $F$. 
The \name{induced graph} of $G$ by $S$ is$ (S, E[S])$ and is denoted by $G[S]$.
$G[F] = (V[F], F)$ is a \name{subgraph} of $G$ by $F$. 
We denote by $G\setminus S = G[V\setminus S]$. 
Since $G[S]$ (resp. $G[F]$) is uniquely determined by $S$ (resp. $F$), 
we identify $S$ with $G[S]$ (resp. $F$ with $G[F]$) if no confusion arises. 

We say that  
a sequence $\pi = (v = w_1, \dots, w_{\ell} = u)$ of vertices in $V$ is a \name{path} of $G$ between $v$ and $u$ 
if for each $i = 1, \dots, \ell-1$, $(w_i, w_{i+1}) \in E$, and 
each vertex in $\pi$ appears exactly once. 
We denote by the \name{length} of a path the number of edges in the path. 
$\pi$ is a cycle if $v = u$ and the length of $\pi$ is at least three. 
The \name{distance} $\dist{u, v}$ between $u$ and $v$ is the 
We say $G$ is \name{connected} if there is a path between any pair of vertices in $G$. 
$G$ is \name{bipartite} if $G$ has no cycle with odd length. 
For a vertex subset $S\subseteq V$ (resp. an edge subset $F \subseteq E$) such that $G[S]$ (resp. $G[F]$) is bipartite, 
we say $S$ (resp. $F$) a \name{bipartite vertex set} (resp. a \name{bipartite edge set}). 
For any bipartite vertex set $S$, 
if $G[S]$ is connected, we say $S$ a \name{connected bipartite vertex set}. 
We also say a bipartite edge set $F$ is a \name{connected bipartite edge set} if $G[F]$ is connected. 
Let $\iSOL{G}$ and $\sSOL{G}$ be the collection of connected bipartite vertex sets and connected bipartite edge sets, respectively. 
We call $\iSOL{G}$ (resp. $\sSOL{G}$) the \name{solution space} for Problem~\ref{prob:ind} (resp. for Problem~\ref{prob:sub}). 
Since we only focus on connected ones, 
we simply call a connected bipartite vertex (resp. edge) set a bipartite vertex (resp. edge) set.  
In what follows, 
we assume that $G$ is connected and simple. 
We now define the enumeration problems of this paper as follows: 

\begin{problem}[Bipartite induced subgraph enumeration]
\label{prob:ind}
For given a graph $G$, output all vertex sets in $\iSOL{G}$ without duplication. 
\end{problem}

\begin{problem}[Bipartite subgraph enumeration]
\label{prob:sub}
For given a graph $G$, output all subgraphs in $\sSOL{G}$ without duplication. 
\end{problem}

\section{Enumeration of Bipartite Induced Subgraphs}

\begin{algorithm}[t]
    \caption{Enumeration algorithm based on binary method}
    \label{alg:proposed}
    \Procedure(){\Main{$G = (V, E)$}} {
        \ForEach{$v \in V$}{
            \Rec{$G, \set{v}, N(v)$}\; 
            $G \gets G\setminus\set{v}$\; 
        }
    }
    \Subprocedure(){\Rec{$G, S, \CAND{S}{G}$}} {
        \Output{$S$}\; 
        \While{$\CAND{S}{G} \neq \emptyset$}{ \label{alg:proposed:main:loop:start}
            $u \gets $ the smallest child generator in $\CAND{S}{G}$\;
            $\CAND{S}{G} \gets \CAND{S}{G} \setminus \set{u}$\;\label{alg:proposed:remove:from:C}
            $S' \gets S \cup \set{u}$\; 
            \Rec{$G, S', \text{\ComputeChildGenerator{$\CAND{S}{G}, u, G$}}$}\; \label{alg:proposed:call:rec}
            $G \gets G \setminus \set{u}$\;\label{alg:proposed:remove:from:G} \label{alg:proposed:main:loop:end}
        }
    }
    \Subprocedure(){\ComputeChildGenerator{$\CAND{S}{G}, u, G$}}{
        \If{$u \in \CANDL{S}{G}$}{
            $\CAND{S\cup\set{u}}{G} \gets \CAND{S}{G}\setminus (\CANDL{S}{G} \cap N(u))$\;
        } \ElseIf{$u \in \CANDR{S}{G}$}{
            $\CAND{S\cup\set{u}}{G} \gets \CAND{S}{G}\setminus (\CANDR{S}{G} \cap N(u))$\;
        }
        $\CAND{S\cup\set{u}}{G} \gets \CAND{S\cup\set{u}}{G} \cup \Addv{S}{u}{G}$\; 
        \Return{$\CAND{S\cup\set{v}}{G}$}\;
    }
\end{algorithm}

In this paper, 
we propose two enumeration algorithms for Problem~\ref{prob:ind} and Problem~\ref{prob:sub}, and this section describes the algorithm for Problem~\ref{prob:ind}. 
The pseudocode of the algorithm is described in Algorithm~\ref{alg:proposed}.
We employ \name{binary partition method} for constructing the algorithms. 
The algorithm outputs the minimal solution to be output, and partition the set of remaining solutions to be output into two or more disjoint subsets. 
Then, the algorithm recursively solve the problems for each subset, by generating recursive calls.
We call this dividing step excluding recursive calls (Line~\ref{alg:proposed:call:rec} in Algorithm~\ref{alg:proposed}) an \name{iteration}. 

For any pair $X$ and $Y$ of iterations, 
$X$ is the \name{parent} of $Y$ if $Y$ is called from $X$ and 
$Y$ is a \name{child} of $X$ if $X$ is the parent of $Y$.

\begin{figure}[t]
    \centering
    \includegraphics[width=0.6\textwidth]{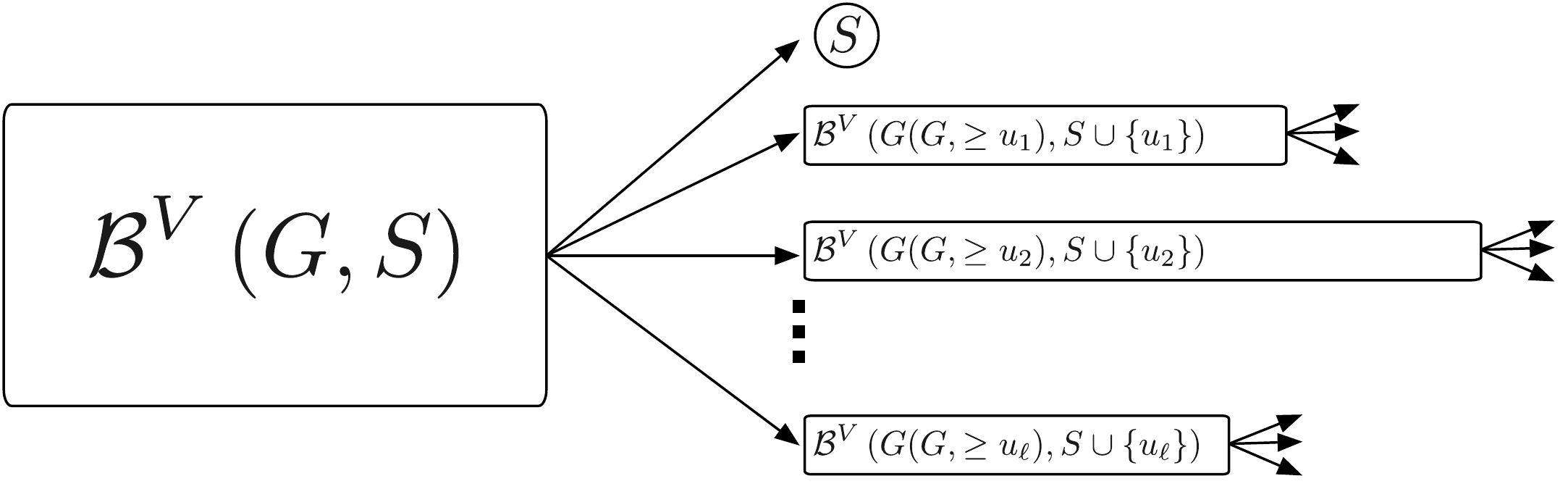}
    \caption{Example of the partitioning the solution space. 
        Algorithm~\ref{alg:proposed} recursively partitions the solution space $\iSOL{G, S}$ 
        into smaller disjoint solution spaces, according to $\CAND{S}{G} = \set{u_1, \dots, u_\ell}$. 
    }
    \label{fig:binary:ind}
\end{figure}

For any bipartite vertex set $S$, 
we say that $S'$ is a \name{child} of $S$ if there exists a vertex $u$ such that  $S' = S\cup\set{u}$. 
A vertex $v \notin S$ is a \name{child generator} of $S$ for $G$ if $S\cup\set{v}$ is a bipartite vertex set in $G$. 
That is,
the proposed algorithm enumerates all bipartite vertex sets by recursively adding a child generator to a current bipartite vertex set $S$. 
We denote by $\CAND{S}{G}$ the set of child generators of $S$ in $G$. 
Suppose that $r$ be the smallest vertex in $S$. 
Let $\Lg{S} = \inset{u\in S}{\dist{u, r} \bmod 2 = 0}$ 
and $\Rg{S} = \inset{u\in S}{\dist{u, r} \bmod 2 = 1}$.  
For any vertex $v$ in $G$, 
any descendant iteration of \Rec{$G, \set{v}, N(v)$} does not output 
a bipartite vertex set including vertices less than $v$. 
Hence, no vertex will never move to the other side in any descendant bipartite vertex set. 
Let $\CANDL{S}{G} = \inset{u\in\CAND{S}{G}}{u \in \Lg{S\cup\set{u}}}$ 
and $\CANDR{S}{G} = \inset{u\in\CAND{S}{G}}{u \in \Rg{S\cup\set{u}}}$.  
Note that $\CAND{S}{G} = \CANDL{S}{G} \sqcup \CANDR{S}{G}$, 
where $A \sqcup B$ is the disjoint union of $A$ and $B$. 
We denote by $\iSOL{G, S} = \inset{S' \in \iSOL{G}}{S \subseteq S'}$ the collection of bipartite vertex sets which include $S$. 
Note that $\iSOL{G} = \iSOL{G, \emptyset}$. 
From now on,
we fix a graph $G$ and a bipartite vertex set $S$ of $G$. 
By the following lemma, 
the algorithm divides $\iSOL{G, S}$ according to $\CAND{S}{G}$ (Fig.~\ref{fig:binary:ind}). 
For an edge $u \in \CAND{S}{G}$, we define $G(S,\ge u)$ by $G \setminus \inset{v \in \CAND{S}{G}}{v < u}$.

\begin{lemma}
    \label{lem:disjoint:solution:space}
$\iSOL{G(S, \ge u), S\cup\set{u}} \cap \iSOL{G(S, \ge v), S\cup\set{v}} = \emptyset$ for any $u \neq v$ of $\CAND{S}{G}$. 
\end{lemma}
\begin{proof}
    Without loss of generality, we can assume that $u < v$. 
    Suppose that $S'$ is a bipartite vertex set in $\iSOL{G(S, \ge u), S\cup\set{u}} \cap \iSOL{G(S, \ge v), S\cup\set{v}}$. 
    This implies that $S'$ includes both $u$ and $v$. 
    However, $G(S, \ge v)$ does not contain $u$. 
    Thus, $S' \notin \iSOL{G(S, \ge v), S\cup\set{v}}$, and this contradicts the assumption. 
    Hence, the statement holds. 
\end{proof}

\begin{lemma}
    \label{lem:divide:solution:sapce}
    $\iSOL{G, S} = \set{S} \cup \bigsqcup_{u \in \CAND{S}{G}}^\ell \iSOL{G(S, \ge u), S\cup\set{u}}$. 
\end{lemma}

\begin{proof}
    If $S' \in \set{S} \cup \bigsqcup_{u \in \CAND{S}{G}}^\ell \iSOL{G(S, \ge u), S\cup\set{u}}$, 
    then $S'$ is obviously in $\iSOL{G, S}$. 
    We assume that $S' \in \iSOL{G, S}$ and $S' \supsetneq S$. 
    This implies that $S'$ includes one of child generators in $\CAND{S}{G}$. 
    Let $v$ be the smallest child generator in $\CAND{S}{G}$ such that $v$ belongs to $S'$. 
    Hence, $S' \subseteq V(G(S, \ge v))$. 
    Therefore, $S' \in \iSOL{G(S, \ge v), S}$ and the statement holds. 
\end{proof}

Next, 
we consider the correctness of \ComputeChildGenerator. 
For brevity, we introduce some notations: 
Let $u$ be a child generator in $\CAND{S}{G}$. 
$\Addv{S}{u}{G} = \inset{w \in N(u)}{w \notin N[S]}$ is the set of vertices that are adjacent to only $u$ in $S\cup\set{u}$. 
Note that $\CAND{S}{G} \cap \Addv{S}{u}{G} = \emptyset$. 
$\Delta(S, G, u) = \CANDL{S}{G} \cap N[u]$ if $u \in \CANDL{S}{G}$;   
$\Delta(S, G, u) = \CANDR{S}{G} \cap N[u]$ if $u \in \CANDR{S}{G}$.  
Intuitively, 
$\Addv{S}{u}{G}$ and $\Delta(S, G, u)$ are the set of vertices that are added to and removed from $\CAND{S}{G}$ to compute $\CAND{S\cup\set{u}}{G(G, \ge u)}$, 
respectively. 
The following lemma shows the sufficient and necessary conditions for computing $\CAND{S\cup\set{u}}{G(G, \ge u)}$.  

\begin{lemma}
    \label{lem:correctness:child:gen}
     $\CAND{S\cup\set{u}}{G(S, \ge u)} 
     = \tuple{\tuple{\CAND{S}{G}\setminus \Delta(S, G, u)} \sqcup \Addv{S}{u}{G}} \setminus \inset{v\in \CAND{S}{G}}{v < u}.$
\end{lemma}
\begin{proof}
    We let 
    $C_* = \tuple{\tuple{\CAND{S}{G}\setminus \Delta(S, G, u)} \sqcup \Addv{S}{u}{G}} \setminus \inset{v\in \CAND{S}{G}}{v < u}$. 
    Suppose that $x \in \CAND{S\cup\set{u}}{G(S, \ge u)}$. 
    Without loss of generality, 
    we can assume that $u \in \CANDL{S}{G}$. 
    From the definition of $G(S, \ge u)$, $x \notin \inset{v\in \CAND{S}{G}}{v < u}$. 
    If $x \notin N[S]$, then  
    since $x$ can be added to $S\cup\set{u}$, 
    $x$ is adjacent to only one vertex $u$ in $S\cup\set{u}$. 
    Hence, $x \in \Addv{S}{u}{G}$.
    If $x \in N[S]$, then
    since $x \in \CAND{S\cup\set{u}}{G(S, \ge u)}$, 
    $x \in \CAND{S}{G}$. 
    Moreover, 
    if $x$ is in $\CANDL{S}{G} \cap N(u)$, 
    then $S\cup\set{u, x}$ has an odd cycle. 
    Hence, the statement holds. 

    Suppose that 
    $x \in C_*$. 
    Without loss of generality, 
    we can assume that $u \in \CANDL{S}{G}$. 
    Since $x \in \CAND{S}{G} \sqcup \Addv{S}{u}{G}$,
    $S' = S\cup\set{u, x}$ is connected. 
    Suppose that $S'$ has an odd cycle $C_o$. 
    Since $S\cup\set{u}$ is bipartite, $C_o$ must contain $x$. 
    This implies that $x$ has neighbors both in $\Lg{S'}$ and $\Rg{S'}$. 
    If $x \in \Addv{S}{u}{G}$, 
    then  $x$ has exactly one neighbor in $S'$ since $u \in \Lg{S\cup\set{u}}$. 
    Hence, $x \in \CAND{S}{G}$. 
    This implies that either (I) $N(x) \cap (S \cup\set{x}) \subseteq \Lg{S}$ or (II) $N(x) \cap (S \cup\set{x}) \subseteq \Rg{S}$. 
    If (I) holds, then 
    $x$ has neighbors only in $\Lg{S\cup\set{u}}$ on $S'$ since $u \in \CANDL{S}{G}$ and $x \in \CANDR{S}{G}$. 
    If (II) holds, then 
    then $x$ has neighbors only in $\Rg{S \cup\set{u}}$ on $S'$ since $x \notin N(u)$. 
    Both cases contradict that $x$ in $C_o$. 
    Hence, $x \in \CAND{S\cup\set{u}}{G(S, \ge u)}$ and the statement holds. 
\end{proof}

From the above discussion, 
we can show the correctness of our algorithm. 

\begin{lemma}
    \label{lem:correctness:algorithm}
    Algorithm~\ref{alg:proposed} correctly enumerates all bipartite vertex sets in $G$. 
\end{lemma}
\begin{proof}
    From Line~\ref{alg:proposed:main:loop:start} to \ref{alg:proposed:main:loop:end}, 
    the algorithm divides the solution space according to $\CAND{S}{G}$ as shown in Lemma~\ref{lem:divide:solution:sapce}. 
    In addition, 
    from Lemma~\ref{lem:correctness:child:gen}, 
    \ComputeChildGenerator correctly computes the sets of child generators of $S'$ 
    where $S'$ is a child of $S$. 
    Hence, the statement holds. 
\end{proof}

\subsection{Update of child generators}

\begin{figure}[t]
    \centering
    \includegraphics[width=0.9\textwidth]{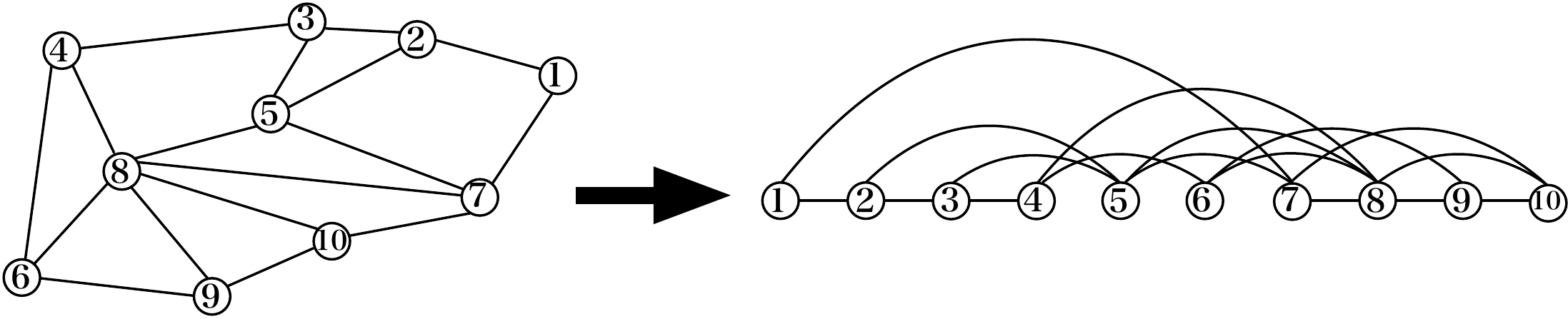}
    \caption{An example of a degeneracy ordering of $G$. 
    The degeneracy of $G$ is two even though there is a vertex with degree six. 
    The right-hand side shows a degenerate ordering of $G$. 
    In the figure, if a vertex $u$ is larger than a vertex $v$, 
    then $u$ is placed at the right $v$. 
    Each vertex has at most two larger neighbors. }
    \label{fig:degenerate}
\end{figure}

In this section, 
we consider the time complexity for the maintenance of the sets of child generators. 
If we na\"{i}vely use Lemma~\ref{lem:correctness:child:gen} for \ComputeChildGenerator, 
we can not achieve $\Order{k}$ amortized time per solution. 
To overcome this, 
we use a \name{degeneracy} ordering on vertices. 
$G$ is a \name{$k$-degenerate} graph~\cite{Lick:White:1970} if 
for any induced graph $S$ of $G$, 
$S$ has a vertex whose degree is at most $k$ (Fig.~\ref{fig:degenerate}). 
The \name{degeneracy} of $G$ is the smallest $k$ such that $G$ is $k$-degenerate.  
Every $k$-degenerate graph $G$ has a \name{degeneracy ordering} on $V$. 
The definition of a degeneracy ordering is that 
for any vertex $v$ in $G$, the number of neighbors of $v$ that are larger than $v$ is at most $k$. 
By recursively removing a vertex with the minimum degree, 
we can obtain this ordering in linear time~\cite{Matula:Beck:1983}. 
Note that there are many degeneracy orderings for a graph. 
In what follows, 
we pick one of degeneracy orderings of $G$ and then 
fix it as the vertex ordering of $G$. 
For any two vertices $u, v$ in $G$, 
we write $u < v$ if $u$ is smaller than $v$ in the ordering. 
We can easily see that 
if $u$ is the smallest child generator, 
then $u$ has at most $k$ neighbors in $G[\CAND{S}{G}]$ 
since $G[\CAND{S}{G}]$ is $k$-degenerate. 
Therefore, 
Lemma~\ref{lem:correctness:child:gen} implies that 
we can compute the child generators of $S\cup\set{u}$ 
by removing at most $k$ vertices 
and adding some vertices that generate some grandchildren of $S$.

\begin{figure}[t]
    \centering
    \includegraphics[width=0.8\textwidth]{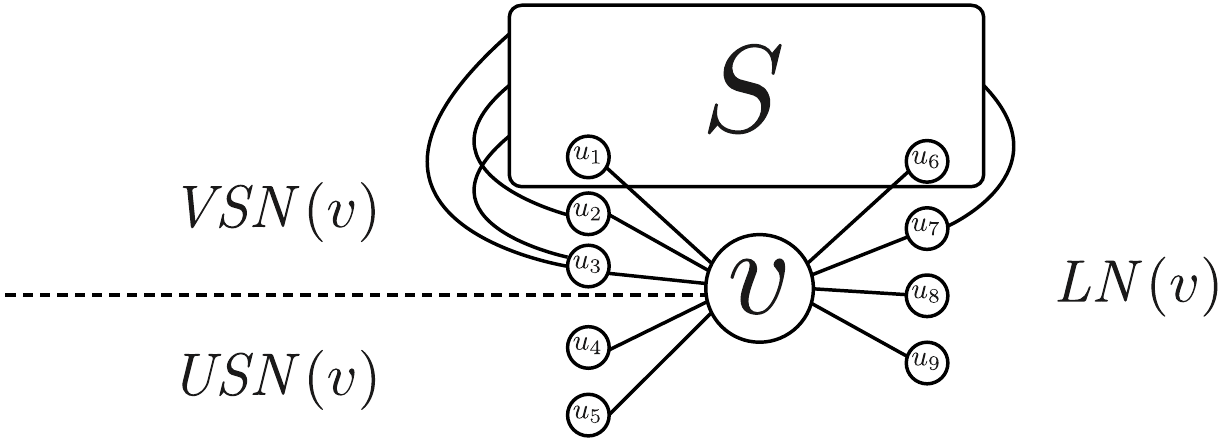}
    \caption{
        Three types of neighbors of $v$. 
        Here, $\VSN{v, S} = \set{u_1, u_2, u_3}$, 
        $\USN{v, S} = \set{u_4, u_5}$, and $\LN{v, S} = \set{u_6, \dots, u_9}$. 
        All of vertices in $\VSN{v, S}$ are adjacent to some vertices in $S$. 
        Thus, for any $S \subseteq S'$, 
        if $u_i$ is in $\VSN{v, S}$, then $u_i$ is also in $\VSN{v, S'}$. 
        In addition, 
        Algorithm~\ref{alg:proposed} also stores $\VSNc{v, S}$ and $\VSNnc{v, S}$. 
    }
    \label{fig:neighbors}
\end{figure}

Next, we define the three types of neighbors of a vertex $v$ in $S$, 
\name{larger neighbors}, 
\name{visited smaller neighbors},
and \name{unvisited smaller neighbors}: 
For any vertex $u \in N(v)$, 
(1) $u$ is a larger neighbor of $v$ if $v < u$, 
(2) $u$ is a visited smaller neighbor of $v$ if $u \in N[S]$ and $u < v$, and
(3) $u$ is an unvisited smaller neighbor otherwise (Fig.~\ref{fig:neighbors}). 
Intuitively, 
$u$ is a visited smaller neighbor if one of its neighbor is already picked in some ancestor iteration of $X$ which receives $S$. 
We denote by $\LN{v, S}$, $\VSN{v, S}$, and $\USN{v, S}$ the sets of 
larger neighbors, 
visited smaller neighbors,
and unvisited smaller neighbors of $v$, respectively. 
In addition, 
the algorithm divides $\VSN{v, S}$ into two disjoint parts $\VSNc{v, S}$ and $\VSNnc{v, S}$; 
$\VSNc{v, S} \subseteq \CAND{S}{G}$ and $\VSNnc{v, S} \cap \CAND{S}{G} = \emptyset$. 
We omit $S$ if no confusion arises.  

We now consider the data structure for the algorithm.
For each vertex $v$, 
the algorithm stores $\LN{v}$, $\VSNc{v}$, $\VSNnc{v}$, and $\USN{v}$ in doubly linked lists. 
$\CAND{S}{G}$ is also stored in a doubly linked list and sorted by the degeneracy ordering. 
The algorithm needs $\Order{m} = \Order{kn}$ space for storing these data structures. 
The algorithm also records the modification when an iteration $X$ calls a child iteration $Y$. 
Let $S_X$ (resp. $S_Y$) be bipartite vertex sets received by $X$ (resp. $Y$). 
Note that for each neighbor $w$ of a vertex $v$, 
if $w$ moves from $\USN{v, S_X}$ to $\VSN{v, S_Y}$, 
then $w$ will never moves from the list in any descendant of $Y$. 
Moreover, 
when $w$ moves to $\VSN{v, S_Y}$, $w$ becomes a child generator of $Y$, 
and thus,  $w \in \VSNc{v, S_Y}$. 
Initially, 
for all smaller neighbors of $v$ is in $\USN{v, \emptyset}$. 
In addition, 
$v$ will be never added to the set in any descendant of $S$
if $v$ is not a child generator of $S$.  
Hence, the algorithm totally needs $\Order{m}$ space for storing the modification history. 
When the algorithm backtracks to $X$ from $Y$, 
the algorithm can completely restore the data structure in the same complexity as the transition from $X$ to $Y$. 
Now, 
we consider the time complexity for the transition from $X$ to $Y$. 
Suppose that when we remove a vertex $v$ from $\CAND{S}{G}$ or add $v$ to $S$, 
for each larger neighbor $w$ of $v$, 
we give a flag which represents $w$ is not a child generator of the child of $S$. 
This can be done in $\Order{k}$ time per vertex because of the degeneracy. 
The next technical lemma shows the number of the larger neighbors which are checked for updating the set of child generators. 

\begin{lemma}
    \label{lem:at:most:one:larger:in:S}
    $S$ has at most one larger neighbor of $v$ 
    for any vertex $v$ in $\CAND{S}{G}$.  
\end{lemma}
\begin{proof}
    Suppose that two or more neighbors of $v$ are in $S$. 
    Let $x$ and $y$ be two of them such that $y$ is added after $x$, 
    and $S' \subseteq S$ be an ancestor bipartite vertex set of $S$ for some graph $G'$ such that $x \in S'$ and $y \notin S'$. 
    Without loss of generality, 
    we can assume that $v$ and $y$ are child generators of $S'$. 
    We can also assume that $v$ is added after $y$. 
    Then, 
    from Lemma~\ref{lem:correctness:child:gen}, 
    When $y$ is added to $S'$, $v$ is not in $G(G', \ge y)$ since $v < y$. 
    This contradicts, and thus, the statement holds.
\end{proof}

\begin{lemma}
    \label{lem:compute:diff:ind}
    Let $u$ and $v$ be two vertices in $\CAND{S}{G}$ such that $u < v$ and $\nexists w \in \CAND{S}{G} (u < w < v)$.  
    $\CAND{S\cup\set{v}}{G(S, \ge v)}$ can be computed from $\CAND{S\cup\set{u}}{G(S, \ge u)}$ in 
    $\Order{k\size{\CAND{S\cup\set{u}}{G(S, \ge u)}} + k\size{\CAND{S\cup\set{v}}{G(S, \ge v)}}}$ time. 
\end{lemma}

\begin{proof}
    From Lemma~\ref{lem:correctness:child:gen}, 
    only the neighbors of $v$ or $u$ may be added to or removed from $\CAND{S\cup\set{u}}{G(S, \ge u)}$ 
    to obtain $\CAND{S\cup\set{v}}{G(S, \ge v)}$. 
    Let $w$ be a vertex in $\LN{v}$. 
    We consider the following cases: 
    (L.1) $w \in \CAND{S\cup\set{u}}{G(S, \ge u)} \cap \CAND{S\cup\set{v}}{G(S, \ge v)}$ or   
     $w \notin \CAND{S\cup\set{u}}{G(S, \ge u)} \cup \CAND{S\cup\set{v}}{G(S, \ge v)}$.  
    In this case, there is nothing to do. 
    (L.2) $w \in \CAND{S\cup\set{u}}{G(S, \ge u)} \setminus \CAND{S\cup\set{v}}{G(S, \ge v)}$. 
    For each larger neighbor $x$ of $w$, 
    we need to move $w$ from
    $\VSNc{x, S\cup\set{u}}$ to $\VSNnc{x, S\cup\set{v}}$. 
    The number such $x$ is at most $k\size{\CAND{S\cup\set{u}}{G(S, \ge u)}}$. 
    (L.3) $w \in \CAND{S\cup\set{v}}{G(S, \ge v)} \setminus \CAND{S\cup\set{u}}{G(S, \ge u)}$. 
    For each larger neighbor $x$ of $w$, 
    we need to move $w$ from $\VSNnc{x, S\cup\set{u}}$ to $\VSNc{x, S\cup\set{v}}$. 
    The number of such $x$ is at most $k\size{\CAND{S\cup\set{v}}{G(S, \ge v)}}$. 
    Note that for each vertex, at most one larger its neighbor is in $S$ from Lemma~\ref{lem:at:most:one:larger:in:S}. 
    Thus, 
    the above three conditions can be checked in constant time for each $w$
    by checking whether or not $w$ is in the same partition as $v$. 
    Therefore, 
    the larger part can be done in $\Order{k + k\size{\CAND{S\cup\set{u}}{G(S, \ge u)}} + k\size{\CAND{S\cup\set{v}}{G(S, \ge v)}}}$ time. 

    Next, let $w$ be a vertex in $\VSNc{v, S}$. 
    From Lemma~\ref{lem:correctness:child:gen}, 
    such $w$ does not belongs to $\CAND{S\cup\set{v}}{G(S, \ge v)}$. 
    Moreover, 
    since $u$ and $v$ are consecutive on $\CAND{S}{G}$, 
    such $w$ is also not in $\CAND{S\cup\set{u}}{G(S, \ge u)}$. 
    Thus, this case can be done in constant time by skipping such vertices. 
    For each vertex $w$ in $\VSNnc{v, S}$,  
    $w$ can not be added to both $S\cup\set{u}$ and $S\cup\set{v}$. 
    Hence, we skip them. 
    In addition,  we need to remove $v$ from $G(S, \ge u)$. 
    This takes $\Order{k}$ time since we only need to update larger neighbors of $v$. 
    The same procedure needs for updating the neighbors of $u$. 
    Hence, the statement holds. 
\end{proof}

Roughly speaking, 
by ignoring neighbors of $u$ or $v$ such that they can not be added to both $S\cup\set{u}$ and $S\cup\set{v}$, 
we can compute $\CAND{S\cup\set{v}}{G(S, \ge v)}$ from $\CAND{S\cup\set{u}}{G(S, \ge u)}$, efficiently. 
In addition, 
other neighbors have corresponding bipartite vertex sets with size $\size{S} + 2$, 
that is, grandchildren of $S$. 
This implies that we can amortize the cost for these neighbors as follows.

\begin{lemma}
    \label{lem:update:child:gen}
    Let $u$ be a vertex in $\CAND{S}{G}$ and 
    $T(S, u)$ be the computation time for $\CAND{S \cup\set{u}}{G(S, \ge u)}$.  
    The total computation time for all the sets of child generators of $S$'s children and recording the modification history is  
    $\sum_{u \in \CAND{S}{G}} T(S, u) = \Order{k\size{\CAND{S}{G}} + \sum_{u \in \CAND{S}{G}}k\size{\CAND{S\cup\set{u}}{G(S, \ge u)}}}$ 
    time. 
\end{lemma}
\begin{proof}
    From Lemma~\ref{lem:correctness:child:gen}, 
    we need $\Order{k\size{\CAND{S}{G}} + k\size{\CAND{S \cup \set{u_*}}{G(S, \ge u_*)}}}$ time for computing $\CAND{S\cup\set{u_*}}{G(S, \ge u_*)}$, 
    where $u_*$ is the smallest child generator in $\CAND{S}{G}$. 
    From Lemma~\ref{lem:compute:diff:ind}, 
    we can compute all the sets of child generators for children of $S$ except for $S\cup\set{u_*}$ 
    in $\Order{\sum_{u \in \CAND{S}{G}}k\size{\CAND{S\cup\set{u}}{G(S, \ge u)}}}$ time in total. 
    Moreover, 
    recording the modification history can be done in the same time complexity in above. 
    Hence, the statement holds. 
\end{proof}

\begin{theorem}
    \label{thm:ind}
    Given a graph $G$ with degeneracy $k$, 
    Algorithm~\ref{alg:proposed} enumerates all solutions in $\Order{k\size{\iSOL{G}}}$ total time, 
    that is,  $\Order{k}$ time per solution with 
    $\Order{m} = \Order{kn}$ space and preprocessing time. 
\end{theorem}
\begin{proof}
    From Lemma~\ref{lem:update:child:gen}, 
     we can see the larger neighbors of $u$ are always checked. 
    Thus, 
    Line~\ref{alg:proposed:remove:from:G} can be done in $\Order{k}$ time since 
    the algorithm does not need to remove edges whose endpoints are $u$ and a smaller neighbor of $u$. 
    Moreover, Line~\ref{alg:proposed:remove:from:C} can be done in $\Order{1}$ time. 
    In addition, 
    in the preprocessing, 
    we need to initialize the data structure and compute the degeneracy ordering. 
    The both need $\Order{kn}$ time and space since the number of edges is at most $kn$. 
    From Lemma~\ref{lem:update:child:gen} and the above this discussion, 
    the algorithm runs in  
    $\Order{\sum_{S\in \iSOL{G}} \sum_{u \in \CAND{S}{G}} T(S, u) }$ time. 
    Now, 
    $\Order{\sum_{S\in \iSOL{G}} \tuple{\size{\CAND{S}{G}} + \sum_{u \in \CAND{S}{G}} \size{\CAND{S\cup\set{u}}{G(S, \ge u)}}}} = 
    \Order{\size{\iSOL{G}}}$. 
    Hence, the statement holds.
\end{proof}

\begin{corollary}
    All bipartite induced subgraphs in graphs with constant degeneracy, such as planar graphs, can be listed in $\Order{1}$ time per solution 
    with $\Order{n}$ space and preprocessing time. 
\end{corollary}

\section{Enumeration of Bipartite Subgraphs}

In this section, 
we describe our algorithm for Problem~\ref{prob:sub}. 
For a graph $G$ and a bipartite edge set $F$ of $G$, 
    let $B(G, F)$ be the set of edges $e$ of $G$ such that $F\cup \set{e}$ is not bipartite,
     i.e., $F\cup \set{e}$ has an odd cycle that includes $e$.
Let $\sSOL{G, F} = \inset{F' \in \sSOL{G}}{F \subseteq F'}$.
We can see that $\sSOL{G(F), F} = \sSOL{G(F)\setminus B(G,F), F}$.
For an edge $e$ of $G$, 
we define 
$\NE{G,e} = \inset{f \in E\setminus\set{e}}{\text{$f$ is adjacent to $e$}}$, 
$\NE{G,F} = \bigcup_{e\in F}\NE{e}\setminus F$, and 
we also define $G(F,\ge e)$ by $G \setminus \inset{ f\in \NE{G,F}}{ f<e }$.

The framework of the algorithm is the same as the algorithm for Problem~\ref{prob:ind}.
The algorithm starts from the empty edge set, 
and add edges recursively so that the edge sets generated are always connected and bipartite, and no duplication occurs. 
For given a graph $G$ and a bipartite edge set $F$ of $G$, 
the algorithm first removes edges of $B(G,F)$ from $G$, and outputs $F$ as a solution.
Then for each $e\in \NE{G, F}$, 
the algorithm generates the problems of enumerating all bipartite subgraphs that include $F\cup \set{e}$ but no edge $f<e, f\in \NE{G, F}$, 
that is, bipartite subgraphs in $G(F,\ge e)$ that includes $F\cup \set{e}$.
Before generating the recursive call, 
the algorithm computes the edges of $B(G(F,\ge e), F\cup \set{e})$ and 
removes them from  $G(F,\ge e)$ so that the computation of the iteration will be accelerated.
The correctness of our strategy for the enumeration is as follows.

\begin{algorithm}[t]
    \caption{Enumeration algorithm for bipartite edges sets}
    \label{alg:proposed:edge}
    \Procedure(){\Main{$G = (V, E)$}} {
        \ForEach(\tcc*[h]{Pick the smallest edge in $E$.}){$e \in E$}{ 
            \Rec{$G, \set{e}, N(e)$}\; 
            $G \gets G \setminus \set{e}$\; 
        }
    }
    \Subprocedure(){\Rec{$G, F, \NE{F, G}$}} {
        \Output{$F$}\; 
        \While{$\NE{F, G} \neq \emptyset$}{
            $e \gets $ the smallest child generator in $\NE{F}{G}$\;
            $F' \gets F\cup\set{e}$\; 
            $G' \gets G(F, \ge e) \setminus E(B(G(F, \ge e), F'))$\; 
            $\NE{G', F'} \gets 
                \tuple{\NE{G,F} \cup N_+(G,F,e)}  
                \setminus \tuple{N_-(G,F,e)\cup \inset{f\in E}{f\le e}}$\; 
            \Rec{$G', F', \NE{G', F'}$}\; 
            $\NE{F, G} \gets \NE{F, G}\setminus\set{e}$\;
        }
    }
\end{algorithm}

\begin{lemma}
    \label{lem:disjoint:solution:sapce:E}
    $\sSOL{G(F,\ge e), F\cup\set{e}} \cap \sSOL{G(F,\ge f), F\cup\set{f}} = \emptyset $ for any $e\ne f$ of $\NE{G,F}$.
\end{lemma}
\begin{proof}
    Suppose that $F'$ is a bipartite edge set in $\sSOL{G(F, \ge e), F\cup\set{e}} \cap \sSOL{G(F, \ge f), F\cup\set{f}}$. 
    This implies $F'$ includes $e_i$ and $e_j$. 
    However, $G_j$ does not include $e_i$, and thus this contradicts the assumption.
    Note that the addition of an edge of $\NE{G,F}$ never results non-bipartite since we removed all edges of $B(G,F)$ from $G$.
\end{proof}

\begin{lemma}
    \label{lem:divide:solution:sapce:E}
    $\sSOL{G, F} = \set{F} \cup \bigsqcup_{e\in \NE{G,F}} \sSOL{G(F,\ge e), F\cup\set{e}}$, 
\end{lemma}
\begin{proof}
    Let $F', F\subset F'$ be a bipartite edge set in $\sSOL{G, F}$, and $e$ be the smallest edge among $F'\cap \NE{G,F}$.
    $e$ always exists since $F'$ is connected.
    We can see that $F'\in \sSOL{G(F,\ge e), F\cup\set{e}}$, thus the statement holds.
\end{proof}

For the efficient computation, our algorithm always keep $\NE{G,F}$ in the memory and update and pass it to the recursive calls.
For the efficient update of $\NE{G,F}$, we keep the graph $G\setminus F$ in the memory since edges of $F$ never be added to $\NE{G,F}$, until the completion of the iteration.
We also put a label of ``1'' or ``2'' to each vertex in $V(F)$ and update so that each edge of $F$ connects vertices of different labels, that is always possible since $F$ is bipartite.
For a vertex $v$ of $G$ that is not in $V(F)$, let $N_1(G, F, v)$ (resp., $N_2(G, F, v)$) be the set of edges $f$ of $N(v)\setminus F$ such that the endpoint of $f$ other than $v$ has label ``1'' (resp., ``2'').
We also keep and update $N_1(G, F, v)$ and $N_2(G, F, v)$.
For an edge $(u,v)\in \NE{G, F}$ such that $u \notin V(F)$, we define $N_+(G, F, (u,v))$ by $N_1(G, F, v)$ and $N_-(G, F, (u,v))$ by $N_2(G, F, v)$ if the label of $v$ 
is ``1'', and $N_+(G, F, (u,v))$ by $N_2(G, F, v)$ and $N_-(G, F, (u,v))$ by $N_1(G, F, v)$ otherwise.
We define $N_1(G, F, (u,v))$ and $N_2(G, F, (u,v))$ by the empty set if both $u$ and $v$ are in $V(F)$

For an edge $e\in \NE{G,F}$, let $F' = F\cup\set{e}$ and $G'$ be the graph obtained from $G(F, \ge e)$ by removing edges of $B(G(F, \ge e), F')$.

\begin{lemma}
    \label{lem:comp:B}
    Suppose that $B(G,F) = \emptyset$ and $e=(u,v)$.
    Then, $B(G', F') = N_-(G, F, v)$.
\end{lemma}
\begin{proof}
    Since $B(G,F) = \emptyset$, any edge $f$ in $B(G', F')$ must share one of its endpoint with $e$, and the endpoint is adjacent to no edge of $F$.
    Further, the edge is not included in $F$.
    The addition of $f$ to $F$ generates an odd cycles if and only if the label of both endpoints of $f$ are the same.
    Therefore the statement holds.
\end{proof}

The following lemma shows that the computation of $\NE{G', F'}$
 from $\NE{G, F}$ can be also done in $\Order{\size{N_+(G, F, e)}+1}$ time.
 
\begin{lemma}
    \label{lem:comp:NE}
    $\NE{G', F'} = \tuple{\NE{G,F} \cup N_+(G,F,e)} \setminus \tuple{N_-(G,F,e)\cup \inset{f\in E}{f\le e}}.$ 
\end{lemma}
\begin{proof}
    We first prove $\NE{G', F'} \subseteq \tuple{\NE{G,F} \cup N_+(G,F,e)} \setminus \tuple{N_-(G,F,e)\cup \inset{f\in E}{f\le e}}$. 
    Let $f$ be an edge in $\NE{G', F'}$. 
    Then, $f$ is in either $\NE{G,F}$ or $N_+(G,F,e)$.
    From the definition of $G(F, \ge e)$ and $B(G', F')$, $f$ is not in $B(G', F')$.
    Further, $f>e$ from $f\in G'$.
    This implies that $f \in \tuple{\NE{G,F} \cup N_+(G,F,e)} \setminus \tuple{N_-(G,F,e)\cup \inset{f\in E}{f\le e}}$. 

    We next prove $\tuple{\NE{G,F} \cup N_+(G,F,e)} \setminus \tuple{N_-(G,F,e)\cup \inset{f\in E}{f\le e}} \supseteq \NE{G', F'}.$
    Suppose that $f \in \tuple{\NE{G,F} \cup N_+(G,F,e)} \setminus \tuple{N_-(G,F,e)\cup \inset{f\in E}{f\le e}}$. 
    Then, $f$ is not in $F'$, and adjacent to an edge of $F'$.
    Further, $f>e$ and the addition of $f$ to $F'$ generates no odd cycle.
    Thus, $f \in \NE{G', F'}$. 
\end{proof}

When we generate $G(F, \ge e)$ for each $e\in \NE{G,F}$ one by one in increasing order, the total computation time is $\Order{ \size{\NE{G,F}} }$.
Computation of $G'\setminus F'$ is at most the time to compute
$G'$.
From this together with these lemmas, we can see that an iteration of the algorithm spends $\Order{\size{\NE{G,F}} + \sum_{e \in \NE{G, F}} \size{N_+(G, F, e)}}$ time.
The following lemma bound this complexity in another way.
Let $G_e$ is the graph obtained from $G(F\cup\set{e}, \ge e)$ by removing edges of $B(G(F\cup\set{e}, \ge e), F\cup\set{e})$.

\begin{lemma}
    \label{lem:time:NE}
    Suppose that $e'$ is next to $e$ in the edge ordering in $\NE{G,F}$. 
    The computation of $\NE{G_{e'}, F\cup\set{e'}}$ from $\NE{G_e, F\cup\set{e}}$ can be done in $\Order{\size{\NE{G_{e'}, F\cup\set{e'}}}+\size{\NE{G_e, F\cup\set{e}}}}$ time.
\end{lemma}
\begin{proof}
The computation is to recover $\NE{G,F}\setminus \inset{f\in E}{f\le e}$ from $\NE{G_e, F\cup\set{e}}$ and construct $\NE{G_{e'}, F\cup\set{e'}}$ from it.
From Lemma~\ref{lem:comp:NE}, its time is linear in 
$\size{N_+(G,F,e)\setminus \inset{f\in E}{f\le e}} + \size{N_-(G,F,e)\setminus \inset{f\in E}{f\le e}} + \size{N_+(G,F,e')\setminus \inset{f\in E}{f\le e}} + \size{N_-(G,F,e')\setminus \inset{f\in E}{f\le e}}$.
We see that $N_+(G,F,e)\setminus \inset{f\in E}{f\le e}\subseteq \NE{G_e, F\cup\set{e}}$, and $N_+(G,F,e')\setminus \inset{f\in E}{f\le e'}\subseteq \NE{G_{e'}, F\cup\set{e'}}$.
When $N_-(G,F,e)\ne N_-(G,F,e')$, we have 
$N_-(G,F,e)\cap N_-(G,F,e') = \emptyset$, thus 
$N_-(G,F,e)\setminus \inset{f\in E}{f\le e}\subseteq \NE{G_{e'}, F\cup\set{e'}}$, and $N_-(G,F,e')\setminus \inset{f\in E}{f\le e'}\subseteq \NE{G_e, F\cup\set{e}}$
thus the statement holds.
When $N_-(G,F,e) = N_-(G,F,e')$, they are canceled out and no need of taking care in the computation, thus the statement also holds.
\end{proof}

\begin{lemma}
    \label{lem:time2:NE}
    For any iteration inputting $G$ and $F$ such that $B(G, F)=\emptyset$, its computation time is at most proportional to one plus the number of its children and the grandchildren.
\end{lemma}
\begin{proof}
For the first recursive call with respect to an edge $e$,
we pay computation time of $\Order{\size{\NE{G, F}} + \size{\NE{G_e, F\cup\set{e}}}}$.
For the remaining recursive calls, as we see in Lemma~\ref{lem:time:NE}, the computation time is linear in the number of grandchildren generated in the recursive call, and that generated just before.
Thus, the statement holds.
\end{proof}

Sine any iteration requires at most $\Order{|V|+|E|}$ space.
When the iteration generates a recursive call, the graphs and variants that the iteration is using has to be recovered, just after the termination of the recursive call.
This can be done just keeping the vertices and edges that are removed to make the input graph of the recursive call.
Thus, the total accumulated space spent by all its ancestors is at most $\Order{|V|+|E|}$.
Therefore, we obtain the following theorem.

\begin{theorem}
    \label{thm:sub}
    All bipartite subgraphs in a graph $G=(V,E)$ can be listed in $\Order{\size{\sSOL{G}}}$ total time, 
    that is, $\Order{1}$ time per solution with 
    $\Order{|V|+|E|}$ space. 
\end{theorem}

\bibliographystyle{abbrv}
\bibliography{wasa}
\end{document}